\documentclass[conference]{IEEEtran}

\usepackage{fix-cm}
\usepackage{cite}

\usepackage{latexsym}
\usepackage{amsmath}
\usepackage{amssymb}
\usepackage{amsfonts}
\usepackage{subfig}
\usepackage{color}
\usepackage{multirow}
\usepackage{adjustbox}
\usepackage{tabularx}
\input{_I.tex}
\input{_II.tex}

\definecolor{gray}{rgb}{0.8,0.8,0.8}
\definecolor{lgray}{rgb}{0.9,0.9,0.9}

\def\ader{\curlyeqprec}
\def\ared{\curlyeqsucc}
\def\bired{\approxeq}
\def\F{\mathcal{F}}
\def\T{\mathcal{T}}
\def\Q{\mathcal{Q}}
\def\R{\mathcal{R}}
\def\t{\tau}
\def\c{\zeta}

\auxfun{val}
\auxfun{expose}
\auxfun{modify}
\auxfun{transformation}
\auxfun{svr}
\auxfun{match}
\auxfun{aftmatch}
\auxfun{policy}
\auxfun{ftmatch}
\auxfun{ftvalue}
\auxfun{aftvalue}
\auxfun{reconcile}

\begin{document}
\title{Exploiting Universal Redundancy}

\author{\IEEEauthorblockN{Ali Shoker}
\IEEEauthorblockA{HASLab, INESC TEC \& University of Minho\\
Braga, Portugal\\
Email: ali.shoker@inesctec.pt}
}

\maketitle
\begin{abstract}
Fault tolerance is essential for building reliable services; however, it comes at
the price of redundancy, mainly the ``replication factor'' and ``diversity''.
With the increasing reliance on Internet-based services, more machines (mainly
servers) are needed to scale out, multiplied with the extra expense of 
replication. This paper revisits the very fundamentals of fault tolerance and 
presents ``artificial redundancy'': a formal generalization of ``exact copy'' 
redundancy in which new sources of redundancy are exploited to build fault tolerant systems. 
On this concept, we show how to build ``artificial replication'' and design
``artificial fault tolerance'' (AFT). 
We discuss the properties of these new techniques
showing that AFT extends current fault tolerant approaches to use other forms of
redundancy aiming at reduced cost and high diversity.

\end{abstract}


\section{Introduction}
\label{sec:intro}
Fault tolerance (FT) is the central pillar of reliable 
services~\cite{Gray:1991:HCS,Gartner:1999:FFD,Chandra:2007:PaxosLive,
Castro:1999:PBFT,Castro:2003:BASE}.
A fault tolerant system must employ some form of redundancy
\emph{Redundancy in space} is often used against a specific pre-defined class 
of faults, like \emph{fail-stop, fail-recovery, Byzantine}~\cite{Gray:1991:HCS,
Castro:1999:PBFT}, and employs a fault tolerant protocol that ensures the 
correctness of the system through replication and agreement techniques, e.g., 
\emph{TMR, consensus, failure detectors}~\cite{Lyons:1962:TMR,Gray:1991:HCS,
Chandra:1996:UnreliableFD}. 
Unfortunately, these techniques are costly and their effectiveness are 
questionable due to the \emph{replication factor} and \emph{diversity}.

In particular, in order to tolerate $f$ 
faults, a FT protocol requires a minimum number of redundant 
components (i.e., replicas), called the \emph{replication factor}, which
is often $2f+1$ and can spans up to $5f+1$ in some 
fault models~\cite{Chandra:2007:PaxosLive,Shvachko:2010:HDFS,
Abd-El-Malek:2005:QU,Kotla:2007:Zyzzyva}. 
The thorough reliance on Internet services nowadays amplifies these costs as
more machines (either virtual of physical) are required to cope with the 
increasing demand on services; due to replication, this cost is multiplied by 
(at least) a factor of two or three.
On the other hand, independence of failures between replicas is a 
major assumption for the correctness of FT protocols: a common fault among 
system replicas can bring them to fail unanimously, rendering the replication 
factor useless~\cite{Knight:1986:NVersionExper,
Eckhardt:1985:TheoryNVersion,Eckhardt:1991:NVersionExper,Castro:2003:BASE}. 
This is usually mitigated by introducing some software and hardware \emph{diversity} 
in the replicated components on different axes and 
levels~\cite{Obelheiro:2006:OSDiversityAxes,Liming:1995:NVersion,
Roeder:2010:ProactiveObfus,Schneider:2004:distrTrust,Castro:2003:BASE,
Hughes:1987:hardwareDiv,Garcia:2011:OSDiverse,Baumann:2009:multikernel}.
Although these approaches improve the FT of 
systems through diversity, they are costly and not always effective~\cite{Knight:1986:NVersionExper,Liming:1995:NVersion,Castro:2003:BASE}. 
For instance, \emph{N-version programming}~\cite{Liming:1995:NVersion,Gray:1991:HCS} 
introduces diversity into software through coding it in multiple versions by 
different independent teams and programming languages. In addition to its high
cost in practice, it has been shown that dependency still exists since (1) developers
often fall in the same mistakes and (2) as long as different 
versions originate from a common specification~\cite{Knight:1986:NVersionExper,
Liming:1995:NVersion}.
Another interesting approach is \emph{proactive recovery} between diverse obfuscated 
components that are generated to be semantically equivalent using a secret 
key~\cite{Roeder:2010:ProactiveObfus, Schneider:2004:distrTrust}; though interesting,
it is only effective in transient failures and when the key is kept 
secret~\cite{Obelheiro:2006:OSDiversityAxes,Cox:2006:N-Variant}.
\emph{BASE}~\cite{Castro:2003:BASE} introduces design diversity using \emph{Components 
Off-The-Shelf} (COTS). The idea is to employ pre-existing diverse components that 
have similar behavior, and then write software wrappers to mimic the state machine 
behavior. Although this approach is promising, it is limited to the existence of COTS in various programming languages.

In this paper, we revisit the very fundamentals of fault tolerance and introduce 
\emph{artificial redundancy} considering the ``redundant information''
inferred through the 
``action on a component'' rather than the ``component'' itself:
a component is artificially redundant to another one if there is a strong 
correlation between them, even if they are non similar in behavior or semantics. 
For example, two always opposite buffers $A=-B$ are artificially redundant
since there is a perfect (though negative) correlation between them, 
as if they are exact copies. 
On the other hand, ``the presence of ice'' is artificially redundant, with some 
\emph{uncertainty}, to ``the atmospheric temperature is low'' since 
they are strongly (but not perfectly) correlated. 

\emph{Artificial replication} can then be achieved by making an artificially 
redundant component an artificial replica, \emph{artira} for short. 
The idea is to wrap the component by an \emph{adapter} to code (resp., decode) the 
input (resp., output) of an artira as needed using component-specific (mathematical 
or probabilistic) transformation functions. 
Adapters are similar to the \emph{conformance wrappers} used in 
BASE~\cite{Castro:2003:BASE}; however, we apply it to completely independent,
but correlated, components instead of those of similar behaviors allowing for
some uncertainty (if needed). 
\emph{Artificial fault tolerance} (AFT) is therefore achieved using replicas and
artira, e.g., using voting or agreement, in a similar fashion to current FT 
protocols. When artira are perfectly correlated, existing FT protocols can
be used with higher reliability due to the increased diversity of artiras.
On the other hand, if artiras include some uncertainty (bounded or unbounded),
new variants of AFT protocols are needed; in Section~\ref{sec:aft}, we only give some
recommendations on how to build such protocols (due to size limitation).

Our work is complementary to existing FT models (in fact, AFT subsumes FT); and it 
is important for three main reasons:
(1) it exploits new forms of redundancy to reduce the cost of replication that can,
optionally, include some uncertainty;
(2) it achieves equivalent or better tolerance to faults than classical FT being built
on highly diverse components in terms of behavior, specification, providers, 
geographic location, etc.; and
(3) it makes it possible to achieve lower levels of fault tolerance, e.g., detection,
if some uncertainty is accepted by the application and when extra ``exact copy'' replicas 
do not exist or are not affordable; or if the service could not be easily replicated
(e.g., redundant medical instruments, large scale social networks, etc.).

AFT may be criticized in two main ways.
First, although AFT achieves high reliability in the case of perfect correlation,
uncertainty in fault tolerance does not seem reasonable in other cases.
Our argument is that, uncertainty even exists (under the hood) in current FT protocols;
in fact, there is no evidence that FT protocols achieve 100\% certainty due to dependence 
of failures~\cite{Eckhardt:1985:TheoryNVersion}.
On the other hand, the uncertainty we introduce here is agreed upon beforehand
(e.g., in the SLA), and thus caution can be planned a priori. In addition, 
uncertainty in fault tolerance also 
do exist in literature as in~\cite{Leners:2013:FailureInformer,
Chandra:1996:UnreliableFD}. The second criticism is how practical is AFT. In fact, similar concepts have been studied and used in specific areas like automotive systems~\cite{ren:2008:distributedConsensus}, clock synchronization~\cite{lamport:1985:synchronizingClocks,mahaney:1985:inexact} and Byzantine approximate agreement~\cite{Dolev:1986:ApproximateAgree}, and we believe that these concepts can be generlized to a wider spectrum of distributed applications and services where even new forms of redundancy can be exploited as we do here.
We show in Sections~\ref{sec:feasibility} and~\ref{sec:apps} that AFT can be applied to 
a large span of applications as in webservices, multithreading, HPC, etc.
On the other hand, several applications can accept some uncertainty in FT 
to reduce the cost of replication or when additional replicas do not 
exist~\cite{Montagnat:2005:MedicalImage,Sarmenta:2002:FTVolunteer}. 

To the best of our knowledge, this subject has not been studied in literature as we 
generlize here, and we believe that it is worth more investigation 
and empirical study in the future. 

\section{Artificial Redundancy}
\label{sec:ared}

\subsection{Notations}
\label{sec:not}
Consider a component $X$ that is associated with a set of possible 
actions in $A$.
In general, an action can modify the state of X; however, for ease of 
presentation, we assume that actions are read-only and we explicitly mention
\emph{writes} when needed.
We denote by $a(x)$ the output of an action $a\in A$ on a state $x\in X$, 
and by $X^a$ the range (i.e., output domain set) of $a$ on any state in X; 
we read this ``X subject to action a''.
Related, $X^A$ denotes the range of any action in A on any state in X, we read 
this ``X subject to actions in A''.
We also assume that $(X^a,d)$ is a \emph{metric space} with a defined distance $d$. 
We argue that any meaningful component output falls in this category.
If $x$ is in a metric space $S$, a ``closed ball'' $N_r(x)$ of center $x$ and 
radius $r$ is the set: $N_r(x)=\{y: y\in S, d(x,y)\leq r\}$; then we say
that $x$ is \emph{in the neighborhood} of y.

\subsection{Artificial Redundancy}
\begin{definition}[Artificial Redundancy]
A component $X$ subject to action $a$ (i.e., $X^a$) is artificially redundant to 
component $Y$ subject to action $b$ (i.e., $Y^b$), denoted $X^a \ared Y^b$, iff there is a 
correlation $\c\in ]0,1]$ between them.
\label{def:artRed}
\end{definition}

The above definition is very relaxed as it makes any two components redundant 
to each other regardless of their behavior or semantics provided that there 
is a correlation between them.
For instance, the atmospheric ``temperature forecast'' component on action 
$\mathsf{getTemperature()}$ is strongly correlated to the 
``snow forecast'' component on action $\mathsf{isFalling()}$, and hence, 
they are artificially redundant.
Although, artificial redundancy often makes sense when there is strong or
perfect correlation (whether +ve or --ve), we do not explicitly 
mention the \emph{strength} of correlation $\c$ to keep the definition general.
The definition does not specify a correlation method to use; however,
in general, any statistical correlation model can be used. 
Definition~\ref{def:artRed} is fine-grained to an individual action 
of a component (e.g., a function in a service API); however, in practice, 
services may not be equivalent (i.e., have different APIs); consequently, 
only parts of a service might be artificially redundant; this
is captured by the following definition.

\begin{definition}[Partial Artificial Redundancy]
A component $X$ subject to actions in $A$ ($X^A$), is partially artificially 
redundant to component $Y$ subject to actions in $B$ ($Y^B$) iff there exist at 
least one action $b\in B$ and a corresponding action $a\in A$ such that $X^a$ 
is artificially redundant to $Y^b$. 
In particular, if this is true $\forall b\in B$ (a \emph{surjection}) then we say 
that $X$ is artificially redundant to $Y$, and denoted by $X^A \ader Y^B$ 
or $X\ared Y$ when A=B.
\label{def:partArtRed}
\end{definition}

Definition~\ref{def:partArtRed} governs the entire component when actions $A$
and $B$ are similar (e.g., if $X$ and $Y$ have the same API) as well as a subset of 
actions of the component (e.g., only some API functions are redundant).
Artificial redundancy is often one-way relation as explained in 
Lemma~\ref{lem:symArtRed}. 

\begin{lemma}[Symmetry of Artificial Redundancy]
If $X^a\ared Y^b$ then $X^a\ader Y^b$, therefore we say $X^a \bired Y^b$. 
To the contrary, if $X^A\ared Y^B$ then $X^A \ader Y^B$ is not necessarily true.
\label{lem:symArtRed}
\end{lemma}
\begin{proof}
First part: If $X^a \ared Y^b$, then there exists a correlation between 
$X^a$ and $Y^b$; since correlation is symmetric by definition, then $X^a\ader Y^b$
follows from Definition~\ref{def:artRed}.
Second part: follows from Definition~\ref{def:partArtRed} since there may 
exist an action $a'\in A$ without any correlation to any action in $B$ 
(i.e., no surjection from $Y^B$ to $X^A$ exists).
\end{proof}

Since the goal of artificial redundancy is to establish (artificial) replication, 
we introduce the following lemma which makes this task easier.

\begin{lemma}[Probabilistic Artificial Redundancy]
If $X^a \ared Y^b$ then $\exists \beta \in [0,1]$ and at least two predicates: 
$\R$ defined on $X^a$ and $\Q$ defined on $Y^b$ such that: $P(\R|Q)\geq \beta$. 
\label{lem:probArtRed}
\end{lemma}
\begin{proof}
	If $X^a \ared Y^b$ with correlation $\zeta$, then $\exists \beta \in [0,1]$ 
	such that for any state $y\in Y^b$ there exists a corresponding state 
	$x\in X^a$ with probability $\geq \beta$. Let the predicate 
	$\R=${``$x$ exists''} and $\Q=${``$y$ exists''}, then the 
	probability that $\R$ occurs given that $\Q$ occurs is $\geq \beta$.
	Therefore, $P(\R|Q)\geq \beta$.
\end{proof}

Lemma~\ref{lem:probArtRed} says that if there are two correlated services, one
can probably find two correlated predicates that are bounded (e.g., by $\beta$).
This lemma will make it easier to establish correlations and lookup new
redundancy forms in the next section.

\section{Artificial Replication}
\label{sec:artira}

\subsection{Definitions}
Artificial redundancy remains useless without the ability to transform artificially
redundant components to replicas that can be used in practice. We make this 
possible by introducing \emph{artificial replication} in Definition~\ref{def:artira}.

\begin{definition}[Artificial Replication]
A component $X^a$ is said to be an artificial replica (\emph{artira} for short) of 
$Y^b$ iff there exists $\alpha \in [0,1]$, $\epsilon \in Y^b$, and a 
\emph{transformation} function $\F:X^a\goes Y^b$ such that 
$\forall y\in Y^b, \exists x\in X^a$ such that
$P\left( \hat x=\F(x)\in N_\epsilon(y)  \right)\geq \alpha$. 
We denote this by: $X^a\in \F_\epsilon^\alpha(Y^b)$.
\label{def:artira}
\end{definition}

Informally, this means that a component X (subject to an action $a$) 
is an \emph{artira} of Y (subject to an action $b$) if we can find a function $\F$ 
such that for every state $y\in Y^b$ there is a state $x\in X^a$ such that $\F(x)$ 
is in the neighborhood of $y$ with some accuracy $\epsilon$ (i.e., the error is bounded) and 
certainty $\alpha$ (i.e., the bound is precise). 
An artira is defined in a triple $(\F,\alpha,\epsilon)$ whose
values must be defined a priori. Notice that, $X$ is an artira of $Y$ means that 
$Y$ is a reference replica and need not to be an artira of $X$.
In principle, $\F$ is used to transform the output $a(x)$ of an action $a\in A$ on 
$x\in X$ to a value $\hat x=\F(a(x))\in Y^b$ such that $\hat x$ is  
close, with distance $\epsilon$, to some $y\in Y^b$ with 
certainty $\alpha$. The two metrics $\alpha$ and $\epsilon$ are strongly 
related and should be adjusted together: $\epsilon=0$ reflects 100\% accuracy of 
$\F$ whereas $\alpha$ tells if this is correct all the time. 
Increasing $\epsilon$ makes the accuracy of $\F$ lower but with better certainty $\alpha$.
We show in Section~\ref{sec:feasibility} how can tuning $\alpha$ and $\epsilon$
bring interesting benefits.

\subsection{Building an Artira}
\label{sec:build}

\begin{lemma}
Artificial redundancy is necessary, but not sufficient, for artificial replication.
\label{lem:necessary}
\end{lemma}
\begin{proof}
Necessary condition: consider the two predicates (i.e., events): $\Q=\{y=v | v \in Y^b\}$ 
and $\R=\{x=u\in {X^a} | \hat x=\F(u)\in N_\epsilon(y)\}$. 
From Definition~\ref{def:artira}, if $X^a$ is an artira of $Y^b$ then 
$\forall y\in Y^b, \exists x\in X^a$ such that 
$P\left( \hat x\in N_\epsilon(y)  \right)\geq \alpha$ for some $\alpha\in[0,1]$; 
hence, if $\Q$ occurs then $P(\R)\geq \alpha$, i.e., $P(\R|Q)\geq \alpha$.
Therefore, $X^a$ is artificially redundant to $Y^b$ as per 
Lemma~\ref{lem:probArtRed}.\\
Sufficient condition: trivial from Definition~\ref{def:artira} since $\F$ does not
always exist.
\end{proof}

\begin{figure}[t]
\centering
\includegraphics[width=.15\textwidth]{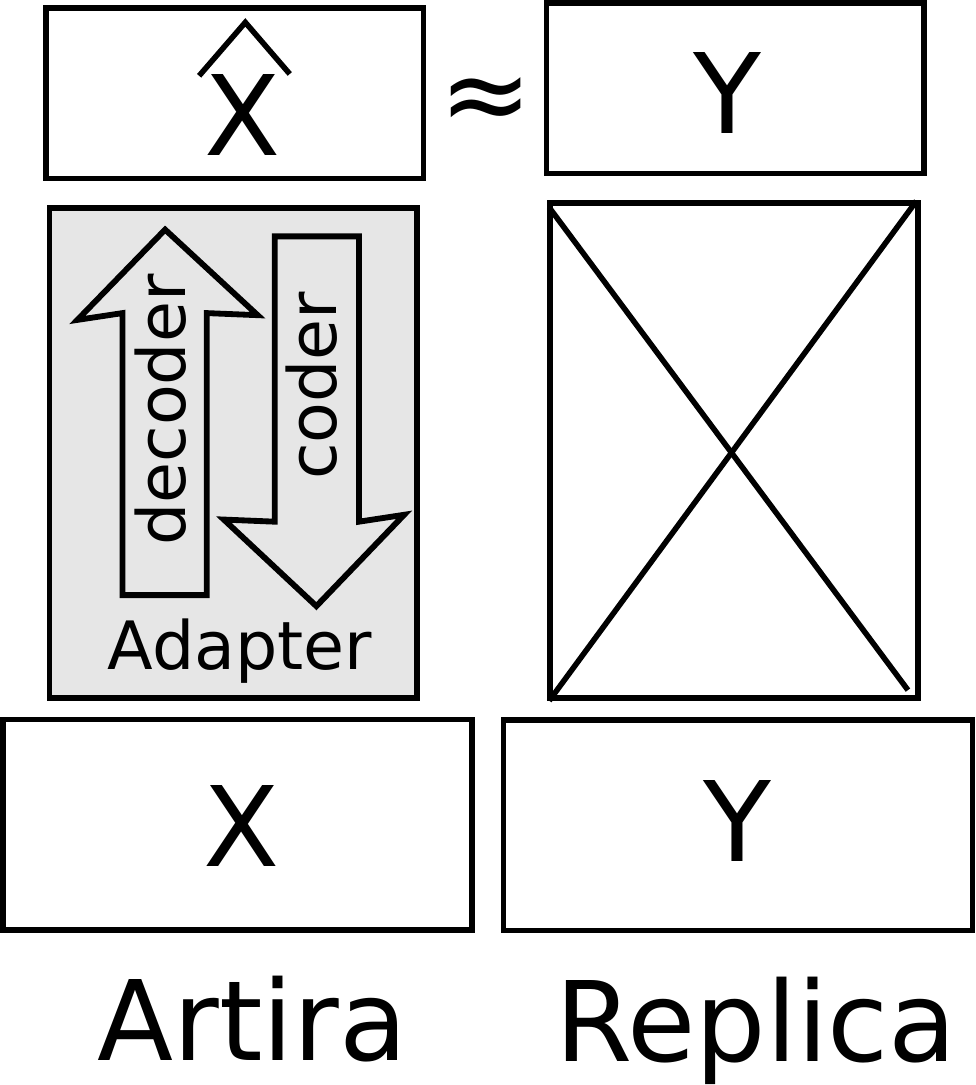}
\caption{Artira vs. Replica.}
\label{fig:artira}
\end{figure}

Lemma~\ref{lem:necessary} means that artiras can be built from artificially redundant 
components if the transformation $\F$ exists.
We argue that such a transformation can often be defined if a strong correlation 
exists either by using accurate mathematical formulas or through using probabilistic
prediction techniques (e.g., using Machine Learning), in which some uncertainty must 
be accepted by the application.

Building an artira $X^a$ of an existing replica $Y^b$ starts by 
defining the accepted accuracy $\epsilon$ and certainty $\alpha$ by the application.
Then, if some strong correlation between $X^a$ and $Y^b$ exists, 
$X^a$ can likely be an artira of $Y^b$. Then, if a transformation $\F$ can be defined
with some accuracy $\epsilon'$ and certainty $\alpha'$, $\epsilon'$ and $\alpha'$ are 
adjusted by incrementing $\epsilon'$
to get a higher certainty $\alpha'$. Finally, $X^a$ is accepted
as an artira of $Y^b$ with the triple $(\F,\alpha', \epsilon')$ if:
$\alpha'\geq \alpha$ and $\epsilon'\leq \epsilon$.
Building an artira $X^A$ of $Y^B$, i.e., considering multiple actions in $A$
and $B$, is done in a similar way by individually defining corresponding triples
as $(\F,\alpha,\epsilon)$.

\textbf{Adapters.} The transformation logic is then implemented in a wrapper on top of $X^a$,
called \emph{adapter}. Fig.~\ref{fig:artira} shows the architecture of an artira
with an adapter versus a replica. An adapter can be state-full or stateless as 
\emph{conformance wrappers} which were explained thoroughly in BASE~\cite{Castro:2003:BASE},
and therefore we skip this discussion here.
\emph{Read} operations use a \emph{decoder} that implements $\F$ to transform 
outgoing values from the artira. \emph{Update} operations, however, use a \emph{coder} 
to write into the artira, which requires an inverse function $\F^{-1}$ to be
defined. In this case, the parameters $\epsilon'$ and $\alpha'$ must be adjusted
to consider the uncertainty of $\F^{-1}$ if it is not a \emph{perfect} inverse of 
$\F$, since a read value will be affected twice by the uncertainty of both $\F$ and 
$\F^{-1}$. However, this is not required when $\F^{-1}$ is a perfect inverse
(e.g., mathematical inverse function) of $\F$, since a value will be read exactly as it
was previously written via the adapter (e.g., if $\F(x)=1/x$ and 
$\F^{-1}(x)=1/x$, then $\F(\F^{-1}(x))=x$). A reasonable cost must be paid while building
an artira as discussed in Section~\ref{sec:aft}.

\section{Artificial Redundancy and Replication Models}
\label{sec:feasibility}

\begin{table}
	\caption{An abstraction of a simple component and the relation between 
two possible components.}
	\begin{tabular}{ll}
	&\\\hline
\textbf{Component:}& $A_i$\\\hline
$\val$:& a value of any type.\\
$\expose(\val)$:& a read-only function that exposes the value of $\val$.\\
$\modify(\val)$:& an update function that modifies the value of $\val$.\\\hline
&\\\hline
\textbf{Relation}& $A_i$ and $A_j$\\\hline
$\t$:&Correlation threshold above which $\c\geq \t$ is accepted.\\
$\T(\val_i,\val_j)$:& a relation upon which a transformation function $\F$ is 
defined.\\\hline
\end{tabular}
\label{fig:compAbst}
\end{table}

We discuss the different artificial redundancy and replication models and their
theoretical feasibility by considering a simple abstraction $A_i$
in Table~\ref{fig:compAbst}. More complex abstractions can intuitively be built
on top of it, but this is enough to serve for explanation.
$A_i$ is composed of a single value $\val$ that represents the state of $A_i$;
whereas $\expose$ and $\modify$ represent the read and write actions, respectively,
that are accessible by any other abstraction $A_j$ (which may have different $\val$ 
type and actions). 
We also represent the relation between $A_i$ and $A_j$ by the correlation threshold 
$\t$ and the relation $\T$, where $\t$ is the minimum correlation coefficient above 
which (inclusive) a service accepts components to be artificially redundant, 
whereas $\T$ materializes the correlation between two components be defining a 
transformation function $\F$ that may comprise some uncertainty as described.
Based on this, we distinguish between interesting artificial replication and 
redundancy models summarized in Table~\ref{tab:models}.

\begin{figure}[t]
\centering
\includegraphics[width=.5\textwidth]{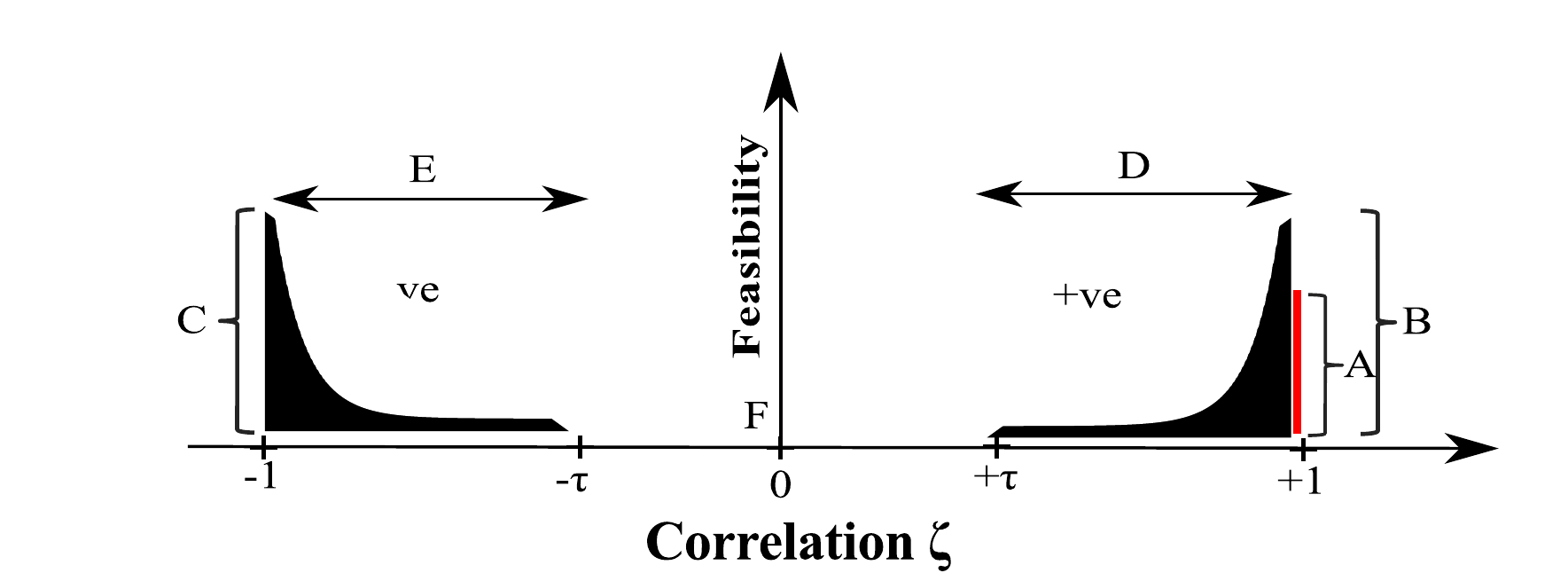}
\begin{tabular}{l|l}
A: Exact copy&	B: Perfect +ve correlation\\
C: Perfect -ve correlation& D: +ve correlation\\
E: -ve correlation&F: No correlation
\end{tabular}
\caption{Feasibility space of artificial redundancy.}
\label{fig:corr-space}
\end{figure}

\begin{table*}[t]
	\caption{The different models, settings, an use-cases of artificial replication
and based on Table~\ref{fig:compAbst}.}
\centering
\begin{tabularx}{\textwidth}{p{.6cm}p{1.3cm}p{1.8cm}X}
\textbf{Model}&\textbf{Correlation}&\textbf{Description}&\textbf{Use-case}\\\hline\hline
\multirow{3}{*}{PAR}&EC&\footnotesize{Exact Copy.}
&\footnotesize{A replicated process having $\T=\{\F(\val_i)=\val_j\}$.}\\\\
&PC+&\footnotesize{Perfect +ve}\newline \footnotesize{Corr.} $\t=+1$&
\footnotesize{Two weather forecast processes where $A_i$ returns
temperature in \emph{Celsius} $\val_i=C^\circ$ and $A_j$ 
in \emph{Fahrenheit} $\val_j=F^\circ$;
then $\T= \{\F(\val_i)=(\val_j -32) \times 5/9\}$.}\\\\
&PC--&\footnotesize{Perfect --ve}\newline \footnotesize{Corr.} $\t=-1$&
\footnotesize{Two processes $A_i$ and $A_j$ having a shared buffer or 
token and using $\val$ as a flag; 
thus $\T=\{\F(\val_i)=-\val_j\}$.}\\\hline
SAR&BSC&\footnotesize{Bounded Strong Corr.} $|\t| \geq 0.5$&
\footnotesize{Two medical diagnosis instruments: cardiac pulse meter $A_i$ and 
Electrocardiogram $A_j$ with sensors $\val_i$ and $\val_j$ (resp.); 
since both monitor heart activity, $\val_i$ and $\val_j$ are strongly 
correlated with some acceptable error $e$ bounded by $\delta$; 
therefore, $\T=\{\F(\val_i)= \val_j\pm e | e\leq \delta\}$.}\\\hline
WAR&USC&\footnotesize{Unbounded Strong Corr.} $|\t| \geq 0.5$& 
\footnotesize{Two social-media profile recommender services that are statistically 
correlated but not bounded due to the tradeoff between accuracy 
and over-fitting. 
If the prediction function has an acceptable unbounded error $e$ and accuracy $p$; 
then $\T=\{\F(\val_i)= \svr(\val_i)= 
\val'_i| P(|\val'_i-\val_j|\leq \pm e)\geq p\}$.}\\\hline
\end{tabularx}
\label{tab:models}
\end{table*}

For ease of presentation, we explain the different models and feasibility with the
help of an ``imaginary'' feasibility spectrum depicted in Fig.~\ref{fig:corr-space}.
Since the ultimate goal is to build fault tolerant systems, which is often the basic
defense layer in a service, critical services tend to adopt very strong
correlation (e.g., $\t$ is close to 1) and, gradually, fewer ones accept lower correlation
coefficients. Therefore, our conjecture argues that the number of applications 
decreases (resp., increases)
\emph{exponentially} to $\t$ (resp., $-\t$) as the correlation coefficient $\c$ approaches 
zero.
Notice that, theoretically, $\t$ can be close to zero; however, it is merely meaningful 
only when $\t>0.5$, i.e., when a strong correlation exists.
Now, we discuss the different models.

\subsection{Perfect Artificial Redundancy and Replication (PAR)}
\label{sec:par}
Perfect artificial redundancy refers to the case in which application FT requirements 
only accept perfect positive correlations between components, i.e., $\t=\pm 1$; 
meaning that the information inferred by one component through the adapter is the exact information
of the other with zero error. Consequently, PAR is the most 
interesting and desirable model to achieve fault tolerance. 
The feasibility is depicted in locations A, B, and C on Fig.~\ref{fig:corr-space}. 
This is mapped to Table~\ref{tab:models}. 
A refers to EC case which is the unique acceptable case in current FT.
Artificial redundancy expands this case to use other redundancy
sources as in PC+ (corr., B) and PC-- (corr., C) with the same confidence as 
if they were exact replicas, as in A.
From the perspective of artificial replication, this case refers to the 
configuration: $(\F,\alpha=1,\epsilon=0)$. As shown in the use-cases of settings 
EC, PC+, and PC--, the function $\F$ transforms $\val_i$ to $\val_j$ 
without any error ($\epsilon=0$) and with 100\% certainty ($\alpha=1$).

\subsection{Strong Artificial Redundancy and Replication (SAR)}
\label{sec:sar}
Strong artificial redundancy (SAR) refers to the case in which a small 
bounded error is tolerable as in BSC case in Table~\ref{tab:models}.
Though SAR is weaker than PAR, it is useful for some applications to 
avoid high costs of exact copies when high certainty is acceptable. 
Of course, such applications are much fewer than those of PAR; 
however, they do exist in practice as 
we explain in Section~\ref{sec:apps}.
In Fig.~\ref{fig:corr-space}, this is depicted in the gradient color regions 
D and E. The dense color indicates more applications, showing that 
the more interesting cases are those when $\tau$ is closer to 1. 
In general, artificial replication is represented by $(\F,\alpha\neq 1,\epsilon\neq 0)$ 
in SAR case; however, the parameters $\alpha$ and $\epsilon$ can be tuned since
the inaccuracy is bounded. Thus, it may be suitable to increase $\epsilon$ so 
that a greater certainty $\alpha=1$ can be achieved  and thus SAR becomes
$(\F,\alpha=1,\epsilon\neq 0)$.
To explain this, consider the use-case in BSC settings in Table~\ref{tab:models}.
In this case, the medical instruments $A_i$ and $A_j$ can infer slightly
different cardiac pulse $\val$ that is bounded by $\delta$. Then, setting 
$\epsilon:=\delta$ such that $\alpha=1$ can be a good choice to get high certainty. 
This actually means that, the artificial replication is 100\% accurate with an allowed 
error of $\delta$ cardiac pulses. We show how this is useful in Section~\ref{sec:aft}.

\subsection{Weak Artificial Redundancy and Replication (WAR)}
\label{sec:war}
Weak artificial redundancy (WAR) refers to USC settings in Table~\ref{tab:models}. 
It is similar to SAR in the correlation threshold $\t$ and in 
feasibility (regions D and E in Fig.~\ref{fig:corr-space}). However,
in WAR, the transformation error could not be bounded, and thus certainty is
never 100\%.
Consequently, weak artificial replication has the configuration: 
$(\F,\alpha\neq 1, \epsilon\neq 0)$ where $\alpha$ and $\epsilon$ could 
not be tuned to have $\alpha=1$ (though tuning is possible in general).
Note that, it is not always required to use artificial redundancy to make decisions, 
but to help making decisions like the use-case in USC settings.
For instance, a fault raised 
by an artira can give an alert to do some ``external'' action (e.g., to explicitly 
verify correctness). Of course, WAR case is not suitable for sensitive applications that
do not accept any uncertainty; however, it can be used in situations that are not
critical to reduce the overhead of replication.

\section{Artificial Fault Tolerance (AFT)}
\label{sec:aft}

\subsection{Definitions}
\begin{definition}[Artificial fault tolerance]
Artificial fault tolerance (AFT) is the approach used to achieve fault tolerance
in a system of redundant components where at least one component is 
artificially redundant.
\label{def:aft}
\end{definition}

Definition~\ref{def:aft} is general as it covers any sort of artificial 
redundancy. The use of redundancy is important since it
is necessary to achieve fault tolerance~\cite{Gartner:1999:FFD}.
Note that requiring at least one artificial redundant component means 
that there may exist other components that
are not necessarily artificial (since artificial redundancy is not 
symmetric, as in Lemma~\ref{lem:symArtRed}). 

In the rest of this paper, we focus on artificial replicas (\emph{artiras}) 
as redundant components since it can be used to build AFT protocols.
Due to size limitations, we address AFT protocols that can be used in active
replication as in fault correction (or fault omission) through 
consensus~\cite{Schneider:1990:StateMachine,Chandra:2007:PaxosLive,Castro:1999:PBFT,
Gray:1991:HCS}. 
This is the dominant 
approach in practice since it gives the user an illusion that no 
faults occurred despite their actual presence. 
The main purpose of the protocols is
to maintain a consistent system state and ensure that a single value is 
reported by the system to the client.
Other forms of fault tolerance, like fault detection, can easily be 
derived from this study.
In the following, we show how current FT protocols can be adjusted to
support artiras, which are the building blocks of AFT protocols. 
We also discuss the properties of the system in this case and the 
differences to FT protocols. 

\subsection{Recall FT Protocols}
\label{sec:ft}

In order to cover a wide range of FT protocols, we focus on the common 
concepts across existing protocols and we explicitly
address any protocol-specific properties. 

Consider a system of $n$ nodes (e.g., replicas) where $f$ of them 
can be faulty (regardless of the fault model). In general, a 
client can send requests to one or more nodes and collect the 
\emph{received} replies from one or more nodes too. 
To ensure correctness, consensus (or agreement) between nodes must be 
achieved. Many FT protocols assume network-partition faults which may split
the nodes into two or more partitions~\cite{Lamport:1998:Paxos,Castro:1999:PBFT}.
Therefore, they only require a \emph{quorum} of nodes to reply correctly. 
To ensure correct Write and 
Read requests, the intersection of a Read quorum and a Write quorum
must be correct (non-faulty). A common approach is to choose the 
quorum to be the majority of nodes (also called majority consensus), 
e.g., $\frac{n}{2}+1$; however, to support other protocols too, 
like 2PC and 3PC~\cite{Bernstein:1987:phaseCommit}, 
we simply refer to this quorum as $q$. A FT protocol achieves 
consensus through three main phases: ~\emph{Propose}, \emph{Accept}, 
and \emph{Learn}. 
\begin{itemize}
	\item Propose: a value is proposed to agree upon.
	\item Accept: a proposed value is accepted by nodes if a 
		quorum $q$ of nodes agree on it after following some
		message exchange pattern.
	\item Learn: the learner (often the requester) accepts the request
		if a quorum $q$ of replies match, and \emph{learns} the matching 
		value.
\end{itemize}

This notation
is analogous to the phases used in the well-known protocols in literature as 
Paxos and BFT~\cite{Lamport:1998:Paxos,Castro:1999:PBFT}; 
however, it also covers other protocols regardless of which node 
is acting as \emph{proposer}, \emph{acceptor}, or \emph{learner}, 
and how these phases are designed. We do not discuss 
message exchange patterns and delivery assumptions of 
an FT protocol since they are often the same as in AFT protocols 
(explained next). Committing a request is also protocol-dependent as it can occur in 
the Accept or Learn phases. 
In some protocols, Reads may not follow these phases 
except for the Learn phase in which a \emph{requester} sends the request 
directly to all nodes and collects enough matching responses. 
In existing FT protocols, the matching logic $\ftmatch$ to approve a request 
by the acceptors and the requester simply requires a quorum $q$ of responses 
$\mathsf{r_k}$ to be equal:
\begin{equation}
	\mathsf{\ftmatch = card(R_q)\geq q;\ R_q=\{i\neq j~\vert~r_i=r_j \}}
	\label{eq:ftmatch}
\end{equation}

Obviously, since all the quorum's responses are equal, the committed value $\ftvalue$
by the acceptors, as well as the learned value by the learner (or requester), 
is a single value which corresponds to any response in the quorum: 
\begin{equation}
	\mathsf{\ftvalue =rand(r_i)\ where\ i \in R_q}
	\label{eq:ftvalue}
\end{equation}

\subsection{Designing AFT Protocols}
Designing an AFT protocol starting from a FT protocol is reasonably not hard 
since the mechanics of the three phases is almost the same. The only
sensitive parts are those which require deterministic behavior.
Since in AFT at least one node will be an artira, this 
can incur some inaccuracy in the response returned by the \emph{decoder} 
(as explained in Section~\ref{sec:artira}) which can induce indeterminism
in some cases. This can require modifications in the three phases depending
on the artificial replication model used.
In general, the phases in an AFT protocol are defined as follows:

\begin{itemize}
	\item Propose: one value is proposed to agree upon.
	\item Accept: \emph{one or more} proposed values are accepted by nodes if a 
		quorum $q$ of nodes agree on them after following some
		message exchange pattern.
	\item Learn: the learner (often the requester) accepts the request
		if a quorum $q$ of replies \emph{match with some uncertainty}, 
		and learns a \emph{chosen value} according to some \emph{policy}.
\end{itemize}

The uncertainty induced by the artiras require different matching logic to 
that in Eq.~\ref{eq:ftmatch} as well since responses may 
not always be equal.
For an AFT model defined by $F_\epsilon^\alpha$, the general 
matching criteria is as follows:
\begin{equation}
	\mathsf{\aftmatch =card(R'_q)\geq q;
		\ R'_q=\{i\neq j|P(d(r_i,r_j)\leq \epsilon)\geq \alpha \}}
	\label{eq:aftmatch}
\end{equation}

Equation~\ref{eq:aftmatch} says that if the distance $d$ (defined in the 
metric space) between two responses is bounded by $\epsilon$ with probability 
$\alpha$ then these responses are considered matching. On the other hand, 
choosing a value by the requester in AFT follows an application-dependent
$\policy$ (e.g., priority, mean value, etc.): 
\begin{equation}
	\mathsf{\aftvalue =\policy(R'_q)}
	\label{eq:aftvalue}
\end{equation}

Next, we discuss the properties of the system and how  
$\aftmatch$ and $\aftvalue$ change according to the replication model used (PAR, 
SAR, or WAR) together with the classical fault models.

\subsubsection{\textbf {Benign Fault Models}}
A fault is considered ``benign'' if the corresponding faulty node
either follows its correct specification or it crashes; examples are:
crash-stop and crash-recovery, as in 2PC, 3PC, and 
Paxos~\cite{Bernstein:1987:phaseCommit,Lamport:1998:Paxos}.
When only benign faults are assumed, we distinguish between the 
following cases:

\paragraph{PAR model}
In the PAR replication model, an AFT protocol has the same design as 
a FT protocol. This is the most desired case
since it is at least as robust as the existing FT case. (Additional
robustness follows from the higher diversity of artiras).
In PAR, the AFT system is
defined by $\F_\epsilon^\alpha$ where $\epsilon=0$ and $\alpha=1$; 
thus, the adapter's coding/decoding is perfect which makes 
the artiras deterministic, and equivalent to replicas in behavior.
Therefore, the matching logic and the learned value become as follows:
\[
	\mathsf{\aftmatch =card(R''_q)\geq q;\ 
		R''_q=\{i\neq j~\vert~d(r_i,r_j)=0 \}}
\]
\[
	\mathsf{\aftvalue =rand(r_i)\ where\ r_i\in R''_q}
\]

Notice that the above equations are exactly equivalent to $\ftmatch$ 
and $\ftvalue$ in Eq.~\ref{eq:aftmatch} and Eq.~\ref{eq:aftvalue}. 
This makes the AFT protocol phases (Propose, Accept, and Learn) 
exactly the same as those defined the FT case in Section~\ref{sec:ft}.
Therefore, in PAR replication model, existing FT protocols can be used.

\paragraph{SAR model}
In the SAR replication model, the AFT system is defined by 
$\F_\epsilon^\alpha$ where $\epsilon\neq 0$ and $\alpha=1$.
This means that the error induced by the adapter's coding/decoding 
is bounded by $\epsilon$ with probability $1$. 
Consequently, the matching logic becomes as follows:
\[
	\mathsf{\aftmatch =card(R_q)\geq q;  
	R'''_q=\{i\neq j~\vert~d(r_i,r_j)\leq \epsilon\}}
\]

Due to this bounded indeterminism, we distinguish between Read and 
Write requests. 
In \textbf{\emph{Write requests}}, the proposer node, in the Propose phase, proposes a 
request value $req_p$. In the Accept phase, a quorum $q$ of nodes 
accept $req_p$ (regardless of the messaging patterns); the
request is then committed by having all non-faulty nodes execute $req_p$ 
in the same order.
However, since artiras are indeterministic, the local state of artiras can 
vary upon execution of $req_p$ within the bound defined by $\epsilon$.
This does not affect the Learn phase since an ACK is enough to
be sent to the requester. 

A pedantic detail is to explicitly check
if the bound $\epsilon$ is respected by having the proposer piggyback
its state to the acceptors. This allows the acceptors to assert
that $\aftmatch$ is satisfied by comparing their local states 
to that of the proposer upon executing the request.
In the protocols where a state value must be piggybacked along with the 
Write response to the requester, two options are possible to choose
the learned value:
\noindent (1) The proposer node that received the request from the requester
can piggyback its state back to the requester after the Accept phase
terminates successfully and the request has been executed by the proposer. 
This however limits the choices of the requester to the state of one node
(the proposer in this case). This might be OK for some applications, but it
might not be desired in others that prefer choosing a value according to 
a certain $\policy$ (e.g., learn a value that corresponds to a replica
instead of an artira). 
\noindent(2) All the states of all nodes are piggybacked to the requester 
in which case the latter can choose the preferred value according to its 
$\policy$. This requires all nodes to share their states in the Accept phase.
In particular,
each node executes the request and replies back to the
proposer or other nodes (depending on the message exchange pattern) 
with its state after execution. 
Since these states are possibly different, 
due to the uncertainty of artiras, the consensus problem is actually 
transformed to a \emph{multi-value} or \emph{Vector consensus}.
Many consensus protocols of this class do exist in literature as 
Interactive Consistency (IC) Vectors~\cite{Pease:1980:Agreement},
Approximate Agreement~\cite{Dolev:1986:ApproximateAgree}, and 
Vector Consensus (VC)~\cite{Correia:2006:VectorConsensus}.
In these approaches, the replicas, and artiras in our case, need to agree on a 
vector (or sub-vector) of correct proposed values instead of a single value.
Using such consensus protocols guarantees that at least a quorum of $q$ values
in the vector $V$ are accepted by acceptors. When $V$ is sent to the requester
in the Learn phase, it can apply its $\policy$ to choose a value from $V$. 

Executing \textbf{\emph{Read requests}} is similar to those of Write requests case.
In some protocols, however, the requester directly sends its Read request to
all nodes, which reply back with their local values to the client, without 
passing through the phases of the protocol described above. 
In this case, the received values can be treated as a vector, and then a 
value is chosen depending on the $\policy$.

A $\policy$ is application-dependent. In some cases, it is enough to choose:
one value randomly, based on some criteria (like $max$ or $min$), or even an 
aggregate value (e.g., $sum$, $mean$). 
We show in Section~\ref{sec:apps} that these policies are sometimes more interesting 
than choosing a single value.

\paragraph{WAR model}
In the WAR replication model, $\epsilon\neq 0$ and $\alpha\neq 1$ since
the error of the adapter coding/decoding could not be bounded. In this
case, the matching policy remains the same as in Eq.~\ref{eq:aftmatch}.
Consequently, it is not recommended to use this model for fault omissions
since there is some probability to report a wrong value to the client.
This can rather be used in fault detection when the application accepts 
some uncertainty and the service could not (or costly to) be easily 
replicated. We do not discuss this case as it is less interesting and 
we only discuss its possible applications in later sections.

\subsubsection{\textbf{Non Benign Faults}}
A fault is said to be ``non benign'' if the corresponding faulty node may not behave
according to its specification even without crashing. Well known examples
are Byzantine and Rational faults~\cite{Castro:1999:PBFT,Aiyer:2005:BAR}. 
Since this can be seen as a form of indeterminism,
we distinguish between the following cases.

\paragraph{PAR model}
In PAR replication model ($F_\epsilon^\alpha$ where $\epsilon\neq 0$ and $\alpha=1$),
an artira behaves exactly as a replica. Therefore, similar to the PAR case
of benign faults above, there are no differences in the design of existing non benign
FT protocols and those of AFT.

\paragraph{SAR model}
In the SAR replication model, the inaccuracy of the adapter coding/decoding 
is bounded by $\epsilon\neq 0$. Therefore, the AFT protocol phases are different from
those of an FT protocol in a similar way to the above discussion of benign FT protocols.
However, since a non benign faulty node can have some 
indeterminism like an artira, this leverages important questions about
how to distinguish between a correct behavior of an artira, 
due to uncertainty, and a faulty replica/artira as defined in the 
fault model like Byzantine, BAR, etc.
~\cite{Lamport:1982:ByzGenerals,Aiyer:2005:BAR}.
This can cause problems under the hood since faulty nodes could skew the results 
of the consensus depending on the $\policy$ used. For instance, if the $\policy=\max$,
then a Byzantine node can always try to reply with the maximum possible value that 
does not violate the bound $\epsilon$. Notice that even if this is within 
the acceptable limits of AFT, it may not be desired. 
We believe that this requires a dedicated research.

\paragraph{WAR model}
This case is similar to that of the benign FT protocols described above.

\section{The Cost of AFT}
\label{sec:cost}
The discussion made in Sections~\ref{sec:intro} and~\ref{sec:aft} shows that
the replication factor raises the cost of fault tolerance to many folds.
This urges the majority of services to use FT techniques that require minimal
number of replicas. Therefore, the norm is often to use the rule $n=2f+1$
and $n=3f+1$ for benign and non benign faults, respectively, and assuming 
$f=1$. Even in this best case, the cost of replication is two or three replicas.
With the growth of computer-based services, the number of servers needed for 
replication will increase faster which imposes more cost on the service owner.
This is not even consistent with the recent sustainability efforts of communities.
The AFT approach we introduce here makes it possible to take advantage of existing
resources to reduce the replication cost. This can stimulate designers
to build their services keeping in mind their potential use as artiras.

However, there is some cost to pay due to: building artiras, new 
protocols, and uncertainty.
For instance, as discussed in Section~\ref{sec:artira}, building artiras requires 
the existence of strongly correlated components. Although this depends on the 
application, it is costly to lookup synergies in a huge space of 
components, and thus it is recommended to designate a fairly small set of 
candidate components to verify correlations. This cost can also be reduced
through using correlation tools as those used in Machine Learning and Data 
Analytics. This tradeoff is also required
while implementing the adapter of an artira which shall not require a large
investment in time and resources; otherwise, the cost will be similar to N-version 
programming~\cite{Liming:1995:NVersion}. Moreover, using artiras from other vendors
can impose additional subscription and accessibility costs as in the cases of 
web-services (refer to Section~\ref{sec:apps}). 
On the other hand, although AFT can use existing FT protocols as in the PAR model,
other models require different protocol variants to consider the induced tolerable
uncertainty. The cost of these protocols can be referred to the 
new forms of agreement (e.g., vector agreement), uncertainty handling, sending full 
replies instead of digests (since replies are not exact), etc. 
We argue that these costs are low compared to the benefits AFT brings in terms of 
using existing components and improving diversity; this also needs a dedicated 
empirical cost study in the future.

\section{Feasibility and Applications}
\label{sec:apps}

\paragraph{AFT Webservices}
An interesting observation is that the web (e.g., web-services, crowdsourcing, 
BigData, etc.)
contains lots of redundant information that is not being used in fault tolerance.
For instance, the leading Web API directory in~\cite{programmableWeb:2015} 
shows that dozens of web-services exist in each API category
(e.g., currency, weather, dictionaries, BigData, etc.). 
Given this, it would be interesting to exploit these redundant sources and use 
them as PAR artiras to design other more reliable service (rewarded by the increased 
diversity and distinct geographic locations of artiras). 
In addition, the SAR model can be used for webservices that can trade some bounded uncertainty 
for reduced cost. Moreover, if using web-services of different behaviors is possible
(e.g., weather forecast and sea level elevation) more reliable webservices 
can be designed using the SAR and WAR models (if statistical correlations exist), 
provided that some uncertainty is accepted. We think that this can introduce 
a new business model in which webservices are provided to clients which in their turn
can build more reliable webservices based on various ones. For instance,
many weather forecast webservices are now available through the SOAP web 
programming model~\cite{Curbera:2002:SOAP}. 
A more robust weather forecast webservice can be built, 
by using other webservices as artiras, and the new service is again sold to clients.
Notice that the overhead of implementing such a webservice is low, using SOAP
for example, and reduces the administrative cost of maintaining the whole
system (since individual artira providers will take care of their webservices).
The above concept can also be applied within the same company. 
For instance, SAR and WAR can be used through using two Google's social media 
services, e.g., Youtube and Google+, as artiras to detect failures in recommender systems, 
e.g., through matching a user profile, interests, preferences, etc. 

\paragraph{AFT in Distributed Programming}
With the introduction of multicore systems, most programming languages 
support multithreading and multiprocessing which significantly increases the 
complexity of software designs. 
This makes software  more prone to errors due to concurrency, shared
resources, memory leaks, etc. Modern programming languages provide methods and 
techniques for error handling and processes recovery. For instance, Erlang 
allows processes to monitor each other for error handling~\cite{cesarini2009erlang}.
If an exception 
is raised by one process, a linked \emph{trapping} process can catch that exception 
and recover the failing process. The exception can comprise useful information 
about the reasons of the crash. The same concept can be used in AFT between 
processes by making some of them artiras to others (e.g., a 
\emph{supervisor} parent process). In this way, processes can raise useful
information that can be redundant to those in other processes even without 
failure and without necessarily using the exception.
Following this approach, AFT can be used among different correlated 
processes and threads that share 
the same resources (drives, tokens, ports, folders etc.), e.g., as in NFS systems.
Another example is to use AFT for Virtual Machines (like JVM) to detect faults 
as in a shared memory prefetching between correlated threads as 
in~\cite{Solihin:2002:MemCorrelation}. 

\paragraph{AFT by Design in HPC}
High Performance Computing (HPC) has recently become a very hot area due to the
high processing power required by Science, Social Network, and in general Big Data.
In HPC, MapReduce is often used to break down a problem into Map and Reduce processes in a 
supercomputer, or simply a large number of processors~\cite{Dean:2008:MapReduce}. 
Maps usually run in parallel
over different processors and Reduce aggregates the final output \emph{emitted} 
by Maps. Since it is costly to replicate the processes for better fault tolerance,
it can be interesting to construct the problem in such a way that different processes
can serve as artiras for others, even if they are 
different~\cite{Schroeder:2010:HPCFailures}. A simple example is the
multiplication of huge matrices in which Map processes are assigned parts of
a matrix like
rows/columns/blocks. If there are patterns in the matrix (e.g., sorting), 
it is not difficult to detect a failure of a Map process if the values emitted by
the adjacent Map processes are captured. Another option is to design the problem
in a way to use processes as artiras, e.g., like using multi-level MapReduce 
steps~\cite{Matrix:2013:HPC}.
In general, \emph{naturally redundant algorithms} can intuitively be split 
into different processes that work in parallel in which they infer some redundant 
information (e.g., computing on the sub-trees of a binary tree) as described
in~\cite{Laranjeira:1991:NaturallyRedundant}. 

\paragraph{Fault Detection}
On the other hand, AFT can be used for fault detection (even in the WAR model). 
In fact, uncertainty and accuracy in failure detectors is well-studied 
in the literature; Failure Informers and Impact FD
~\cite{Leners:2013:FailureInformer,Chandra:1996:UnreliableFD} 
are few examples. These works support our idea that fault tolerance
can accept some uncertainty as in SAR and WAR models. 
Moreover, SAR and WAR can also be used for diagnosis. 
The BSC use-case in Table~\ref{tab:models} is an example for detecting failed medical 
instruments. Similar applications can also accept some error like medical 
imaging~\cite{Montagnat:2005:MedicalImage}, video streaming, self correcting applications
(as genetic algorithms).

\section{Related Work}
\label{sec:related}

Von Neumann introduced the idea of redundancy in software by showing how
reliable ``organisms'' can be synthesized from unreliable components via majority 
``organs''~\cite{VonNeumann:1954ReliableOrganisms}. Most modern fault tolerance 
approaches used this concept through designing protocols with Active Replication as 
in~\cite{Gray:1991:HCS,Schneider:1990:StateMachine,Castro:1999:PBFT,
Chandra:2007:PaxosLive}. In order to tolerate $f$, faults, a number of extra components 
(i.e., replication factor) is required, based on $f$, with the assumption that replicas fail
independently. In this paper, the AFT concept is based on artificial redundancy
where replicas may not be exact-copies. This allows for other relaxed 
FT models that can be useful to reduce the cost of replication.

Fault tolerance protocols assume that replicas fail independently. 
To improve independence of failures, different forms of diversity can be introduced
across different components on different axes and 
levels~\cite{Obelheiro:2006:OSDiversityAxes}:
N-version programming~\cite{Liming:1995:NVersion},
obfuscated software components~\cite{Roeder:2010:ProactiveObfus},
Components off-the-shelf (COTS)~\cite{Castro:2003:BASE},
various hardware~\cite{Hughes:1987:hardwareDiv}, 
diverse OS~\cite{Garcia:2011:OSDiverse}, 
instruction sets~\cite{Kc:2003:RandomInstrSet}, 
memory obfuscation~\cite{Xu:2003:MemRandom}, 
execution and compilation~\cite{Cox:2006:N-Variant}, etc. 
The most popular among these are N-version 
programming, ``proactive recovery'', and using COTS to include design diversity to the
system. N-version programming requires implementing a software by different independent
teams which is very expensive~\cite{Knight:1986:NVersionExper,
Eckhardt:1991:NVersionExper}. The second approach uses
semantically equivalent generated (using a secret key) obfuscated components in 
proactive recovery fashion; this approach is only effective in ``fast'' failures 
and when the key is kept secret as described in~\cite{Obelheiro:2006:OSDiversityAxes,
Cox:2006:N-Variant}. Finally, using COTS, as in BASE~\cite{Castro:2003:BASE} is a 
promising approach as it uses already existing components of similar behaviors,
which reduces the diversity cost (as in N-version programming). 
The approach uses \emph{conformance wrappers} for components to achieve a
similar behavior to state machines~\cite{Schneider:1990:StateMachine}. 
Our work generalizes this idea to use components of different behaviors, 
which has a larger application space and improved diversity. 
In addition, the logic and design of \emph{adapters} in this paper is very similar to
\emph{conformance wrappers} in BASE.

Our idea is close to using diverse assistant systems in automotive systems~\cite{ren:2008:distributedConsensus} where a different system can take over when a critical system fails; however, we generalize this to span generic redundant system that can be used in distributed systems. On the other hand, \emph{clock synchronization}~\cite{lamport:1985:synchronizingClocks,mahaney:1985:inexact} and \emph{Byzantine approximate agreement}~\cite{Dolev:1986:ApproximateAgree} tried to strengthen the validity property through choosing approximate or average values. However, these works addressed systems where replicas (or clocks) are identical. Our work exploits the case where distributed componenets are different in behaviour and semantics; it looks up new redundancy forms even through opposite components in nature and functionallity, which ensures validity through using the ``adapters''.

\section{Concluding Remarks}
\label{sec:con}
This paper introduces a new form of \emph{artificial redundancy}
that is based on the correlations among components rather than on
exact copies or similar behaviors. This allows to exploit new sorts of
redundancy aiming at reducing the cost of replication and improving 
independence of failures. We described how to build artificial replicas,
i.e., artiras, using artificial redundancy and explained the possible
models depending on the correlations. Then, artificial fault tolerance (AFT)
was introduced based on artiras instead of replicas, which induces 
some uncertainty in some cases.
Our work is complementary to previous FT approaches and obviously does not
replace them. AFT PAR model allows to use artiras in a similar fashion 
and certainty as conventional FT models use replicas,
but with higher tolerance to faults due to the increased diversity by artiras. Other
AFT SAR and WAR models can be used if the application tolerates some degree
of uncertainty. In such cases, conventional FT protocols must be adjusted to
include uncertainty. This imposes some cost as explained in Section~\ref{sec:aft}. 
We discussed the feasibility and applications (refer to 
Section~\ref{sec:apps}) of our approach and explained that new business models can
even be introduced on this concept.
We beleive that is interesting to empirically study the tradeoffs between FT and AFT in terms
of fault tolerance, efficiency, and cost in the future.

\bibliographystyle{IEEEtran}
\bibliography{IEEEAbrv,ref}

\end{document}